\documentclass[conference]{IEEEtran}
\onecolumn
\usepackage{blindtext}
\usepackage{graphicx}
\usepackage{amsthm}
\usepackage{amssymb}
\usepackage{amsmath}
\usepackage{comment}
\usepackage{xcolor}
\usepackage{pgfplots}
\pgfplotsset{compat=newest}
  \usetikzlibrary{plotmarks}
  \usetikzlibrary{arrows.meta}
  \usepgfplotslibrary{patchplots}
  \usepackage{grffile}
  \usepackage{amsmath}

\newtheorem{definition}{Definition}
\newtheorem{lemma}{Lemma}
\newtheorem{theorem}{Theorem}
\newtheorem{corollary}{Corollary}
\newtheorem{assumption}{Assumption}

\begin{document}

\title{Mismatched Guesswork\vspace{-.25in}}

\author{{\bf Salman Salamatian,
        Litian Liu,
        Ahmad Beirami,
        Muriel M\'edard}\\
        Research Laboratory of Electronics, MIT, Cambridge, MA, USA.\\
        Email: \{salmansa, litianl, beirami, medard\}@mit.edu\vspace{0.1in}
        }

\maketitle

\begin{abstract}
We study the problem of mismatched guesswork, where we evaluate the number of symbols $y \in \mathcal{Y}$ which have higher likelihood than $X \sim \mu$ according to a mismatched distribution $\nu$.
We discuss the role of the tilted/exponential families of the source distribution $\mu$ and of the mismatched distribution $\nu$. 
We show that the value of guesswork can be characterized using the tilted family of the mismatched distribution $\nu$, while the probability of guessing is characterized by an exponential family which passes through $\mu$. Using this characterization, we demonstrate that the mismatched guesswork follows a large deviation principle (LDP), where the rate function is described implicitly using information theoretic quantities. We apply these results to one-to-one source coding (without prefix free constraint) to obtain the cost of mismatch in terms of average codeword length. 
We show that the cost of mismatch in one-to-one codes is no larger than that of the prefix-free codes, i.e., $D(\mu\| \nu)$.
Further, the cost of mismatch vanishes if and only if $\nu$ lies on the tilted family of the true distribution $\mu$, which is in stark contrast to the prefix-free codes. These results imply that one-to-one codes are inherently
 more robust to mismatch.  
\end{abstract}

\IEEEpeerreviewmaketitle

\section{Introduction}
Consider a random variable $X$ drawn from the distribution $\mu$ supported on a finite alphabet $\mathcal{X}$. Guesswork, denoted by $G_\mu(X)$, is defined as the number of symbols in $\mathcal{X}$ whose likelihood exceeds $\mu(X)$. 
Guesswork was first studied to derive lower bounds on the computational cost of sequential decoding with the aim of characterizing the cutoff rate~\cite{massey94}. Ar{\i}kan derived bounds on guesswork in terms of the R{\'e}nyi entropy of the distribution $\mu$ and the size of the support $\mathcal{X}$~\cite{Arikan}.

Guesswork quantifies the computational cost of sequential decoding~\cite{massey94, Arikan}, computational security against brute-force attack~\cite{christiansen2015multi, Beirami-ISIT15}, the probability of error in list decoding~\cite{MerhavEr2,MerhavEr4}, and the length of the code in one-to-one source coding~\cite{Beirami-ISIT15,Kosut}. 
Ar{\i}kan and Merhav studied guesswork subject to an allowable distortion~\cite{ArikanMerhav}.
Sundaresan studied guesswork subject to source uncertainty~\cite{Sundaresan}. Hanawal and Sundaresan also studied universal guesswork where the source distribution is unknown~\cite{HanawalSundaresan}. Beirami {\em et al.} further derived an individual sequence version of universal guesswork~\cite{Beirami-ISIT15}. Merhav and Cohen~\cite{asaf-universal-guesswork}  extended the universal guessing setup to randomized guessing.
Christiansen and Duffy proved that guesswork satisfies a large deviation principle~\cite{ChristiansenDuffy}.
Beirami {\em et al.} provided an implicit characterization of the large deviations behavior of guesswork through information theoretic quantities~\cite{Beirami-IT}. The study of guesswork has also been extended to the multi-user setting~\cite{christiansen2015multi}, the distributed setting~\cite{salamatian-isit17}, guessing subject to constraints~\cite{Beirami-ISIT15,Rezaee-guesswork}, and guessing with limited memory~\cite{Salamatian-TIFS}.

In this paper, we study the probabilistic behavior of the so-called mismatched guesswork $G_\nu(X)$, for $X \sim \mu$, and $\nu$ is a mismatched distribution $\nu \neq \mu$. In many of the applications discussed above, mismatch is inevitable in practice, as the source distribution is usually obtained via a sample estimation, which is prone to imprecision. We study the large deviations of mismatched guesswork building on a framework involving tilted distributions which were first introduced in \cite{Beirami-allerton15} and then expanded in \cite{Beirami-IT}. 
We consider the case where a sequence of length $n$ denoted by $x^n$ is drawn i.i.d. from $\mu$, while the mismatched distribution is the product distribution of $\nu$, denoted $\nu^n$.
We prove that, on the one hand, $G_{\nu^n}(x^n)$ is related to the entropy of the ``projection'' of the type of $x^n$ on the tilted family of the mismatched distribution $\nu$. On the other hand, the probability of a sequence $x^n$ is related to the KL-divergence of its type with the true distribution $\mu$. These two observations form the basis of our analysis in this paper.
 
We also explore the application of mismatched guesswork in one-to-one source coding, i.e., source
coding without the prefix constraint. Mismatched guesswork has a direct application in this setting, and is the counterpart of the usual mismatch prefix-free source coding. 
It is well known that, in contrast to the prefix-free source codes, the average length of the one-to-one source codes converge to the entropy rate from below at a rate $-1/2 \log (n) / n$ when the distribution is matched~\cite{szpankowski-2008}.
It was also shown that the cost of universality is smaller in one-to-one codes because of one less degrees of freedom~\cite{Kosut,beirami-itw14,Beirami-IT}.
To complete the characterization, we show that one-to-one source codes are more robust to an incorrect knowledge of the source distribution. Moreover, it is possible to obtain the exact same optimal performance of an optimal one-to-one encoder with a mismatched distribution $\nu$, under the condition that $\nu$ is on the tilted family of the true distribution $\mu$. 

The rest of the paper is organized as follows. In Section~\ref{sec:background}, we introduce the notation and the geometric lemmas used in the paper.  Section~\ref{sec:matched} contains a brief summary of the main results for matched guesswork, namely the LDP and the asymptotic growth rate. In Section~\ref{sec:mismatched}, we provide the main results of this paper and characterize the LDP rate function and growth-rate for the mismatched guesswork. Section~\ref{sec:one_to_one} is an application of the main results to one-to-one source coding. We end the paper with the concluding remarks in Section~\ref{sec:conclusion}.

\section{Background} \label{sec:background}
\subsection{Notation}

Let $\mu$ be a distribution on a finite alphabet $\mathcal{X}$.
We denote random variables by uppercase letters, e.g. $X$, and realizations of these random variables with lowercase letters, e.g. $x$. The simplex of all distributions over the alphabet $\mathcal{X}$ is denoted by $\Delta_{\mathcal{X}}$.
We let $H(\mu)$ be the entropy under $\mu$, i.e. $H(\mu) = \sum_{x \in \mathcal{X}} \mu(x) \log \frac{1}{\mu(x)}$ \footnote{In this paper, all logarithms are measured in nats.}.
 
For two distributions $\mu$ and $\nu$ such that $\nu$ is absolutely continuous with respect to $\mu$, the relative entropy (KL-divergence) is defined as usual $D(\mu \| \nu) = \sum_{x \in \mathcal{X}} \mu(x) \log \frac{\mu(x)}{\nu(x)}$. We let $H(\mu \| \nu) \triangleq H(\mu) + D(\mu \| \nu)$ denote the cross entropy. The R\'enyi entropy of order $\alpha$ is defined as:
\begin{align}
    H_{\alpha}(\mu) = \frac{1}{1 - \alpha} \log \sum_{x \in \mathcal{X}} \mu(x)^\alpha.
\end{align}
For sequences, we use a superscript to denote the length, e.g. $x^n$ is a sequence of length $n$. 
We also let $\mu^n$ denote the $n$-fold product distribution of $\mu$, therefore $X^n \sim \mu^n$ means that the sequence of random variables $X^n$ is generated i.i.d. from $\mu$.
We denote by $\mathbf{q}_{x^n}$ the type or empirical distribution of a sequence, and by $\Gamma_n$ the set of possible types for sequences of length $n$.
For any set $\cal C$ we use notations $\mathrm{int} \mathcal{C}$ and $\mathrm{cl} \mathcal{C}$ to denote the interior and the closure of set $\cal C$ respectively.
Finally, the uniform distribution on $\mathcal{X}$ is denoted by $\mathbf{u}_\mathcal{X}$.

\subsection{Geometry}
We make the following two assumptions on all the probability distributions that we study in this paper:
\begin{assumption}
We say $\mu$ is unambiguous if it satisfies the following:
\begin{enumerate}
    \item $\mu(x) > 0$ for all $x \in \mathcal{X}$.
    \item $\underset{x\in \mathcal{X}}{\mathrm{argmin}}\; \mu(x)$ and $\underset{x\in \mathcal{X}}{\mathrm{argmax}}\; \mu(x)$ are unique.
\end{enumerate}
\end{assumption}

Note that the set of distributions which are not unambiguous forms a set of Lebesgue measure zero in the set of all distributions, which can be seen by the fact that the non-umanbiguous distributions are contained in a finite union of lower dimensional sets. See~\cite{Beirami-IT} for the implications of this assumption.

We are ready to define the tilt.
\begin{definition}[mismatched tilt]
Let $\alpha \in \mathbb{R}$  and $\nu$ be unambiguous. We denote by $T(\mu,\nu,\alpha)$ the mismatched tilted distribution of order $\alpha$ of $\nu$ with respect to $\mu$, defined as
\begin{align}
    [T(\nu, \mu, \alpha)](x_i) \triangleq \frac{\mu(x_i) \cdot \nu(x_i)^\alpha}{\sum_{x \in \mathcal{X}} \mu(x) \cdot \nu(x)^\alpha}. \label{eq:def_tilted_nu_mu}
\end{align}
We further define the the mismatched tilted family of $\nu$ with respect to $\mu$ as
\begin{align}
    \mathcal{T}_{\nu, \mu} &\triangleq \left\{ T(\nu, \mu,\alpha) : \alpha  \in \mathbb{R}\right\}.
\end{align}
\end{definition}
\noindent By taking limits, we define :
\begin{align}
    &[T(\nu, \mu, \infty)](x) = \left\{ \begin{array}{cl}
        1 & \text{if } x = \mathrm{argmax}_{x \in \mathcal{X}} \; \nu(x) ,\\
        0 & \text{otherwise}
    \end{array} \right., \\
    &[T(\nu, \mu, - \infty)](x) = \left\{ \begin{array}{cl}
        1 & \text{if } x = \mathrm{argmin}_{x \in \mathcal{X}} \; \nu(x) ,\\
        0 & \text{otherwise}
    \end{array} \right. ,\\
    &T(\nu, \mu, 0) = \mu.
\end{align}
This definition of mismatched tilt generalizes the tilt defined in~\cite[Definition 13]{Beirami-IT}, and recovers it when  $\mu$ is the uniform distribution.
The tilted family $\mathcal{T}_{\nu, \mathbf{u}_\mathcal{X}}$, is denoted by $\mathcal{T}_{\nu}$, and  $T(\nu,\mathbf{u}_\mathcal{X},\alpha)$ is denoted by $T(\nu,\alpha)$. Further, define $\mathcal{T}_\nu^+ = \left\{ T(\nu,\alpha) : \alpha > 0 \right\}$ as the positive tilted family, and $\mathcal{T}_\nu^- = \left\{ T(\nu,\alpha): \alpha < 0\right\}$ as the negative tilted family. Note that $\mathcal{T}_\nu = \mathcal{T}_\nu^+ \cup \mathcal{T}_\nu^- \cup \mathbf{u}_\mathcal{X}$.

\begin{lemma}[closure of the tilted family under tilt operation]
For any $\alpha>0$, the following holds:
\begin{equation}
    \mathcal{T}_{\nu, \mu} = \mathcal{T}_{T(\nu, \alpha), \mu}.
\end{equation}
\end{lemma}
\begin{proof}
    The proof follows from the definition of $T(\nu,\alpha)$ and from \eqref{eq:def_tilted_nu_mu}.
\end{proof}
We now define a collection of linear families.
\begin{definition}[linear family] We denote by $\mathcal{L}(\nu, \alpha)$ the linear family of $\nu$ of order $\alpha$, defined as
\begin{align}
\mathcal{L}(\nu, \alpha) \triangleq \{ \gamma \in \Delta_{\mathcal{X}} : H(\gamma \| \nu) = H(T(\nu, \alpha) \| \nu) )\}
\end{align}
\end{definition}
Intuitively, the mismatched tilted family $\mathcal{T}_{\nu,\mu}$ and the tilted family $\mathcal{T}_{\nu}$, correspond to the curves that are \emph{orthogonal} to the linear families $\mathcal{L}(\nu,\alpha)$, and pass through $\mu$ and $\mathbf{u}_\mathcal{X}$, respectively. We refer the interested reader to \cite[Section 3]{csiszar2004information} for an overview of the duality between linear and exponential families, and their applications in statistics, information theory, and large deviations theory.

For a distribution $\mu$, we can also define projections on a tilted family $\mathcal{T}_\nu$ in the following way:
\begin{definition}[projection on a tilted family]\label{def:projection}
We say $\Pi_{\mathcal{T}_\nu}(\mu)$ is the projection of $\mu$ on $\mathcal{T}_\nu$ and define it as
\begin{align}
    \Pi_{\mathcal{T}_\nu}(\mu) &\triangleq \arg_{\gamma \in \mathcal{T}_\nu } \left\{ H(\gamma \| \nu) = H(\mu \| \nu) \right\} . 
     \label{eq:projection}
\end{align}
\end{definition}
\noindent Note that $\Pi_{\mathcal{T}_\nu}(\mu) = T_\nu \cap \mathcal{L}(\nu, \alpha^\star) = T(\nu, \alpha^\star)$ with $\alpha^\star$ selected such that $H(T(\nu,\alpha^\star)\|\nu)=H(\mu \| \nu)$.

\begin{figure}
    \centering
    \includegraphics[scale=.45]{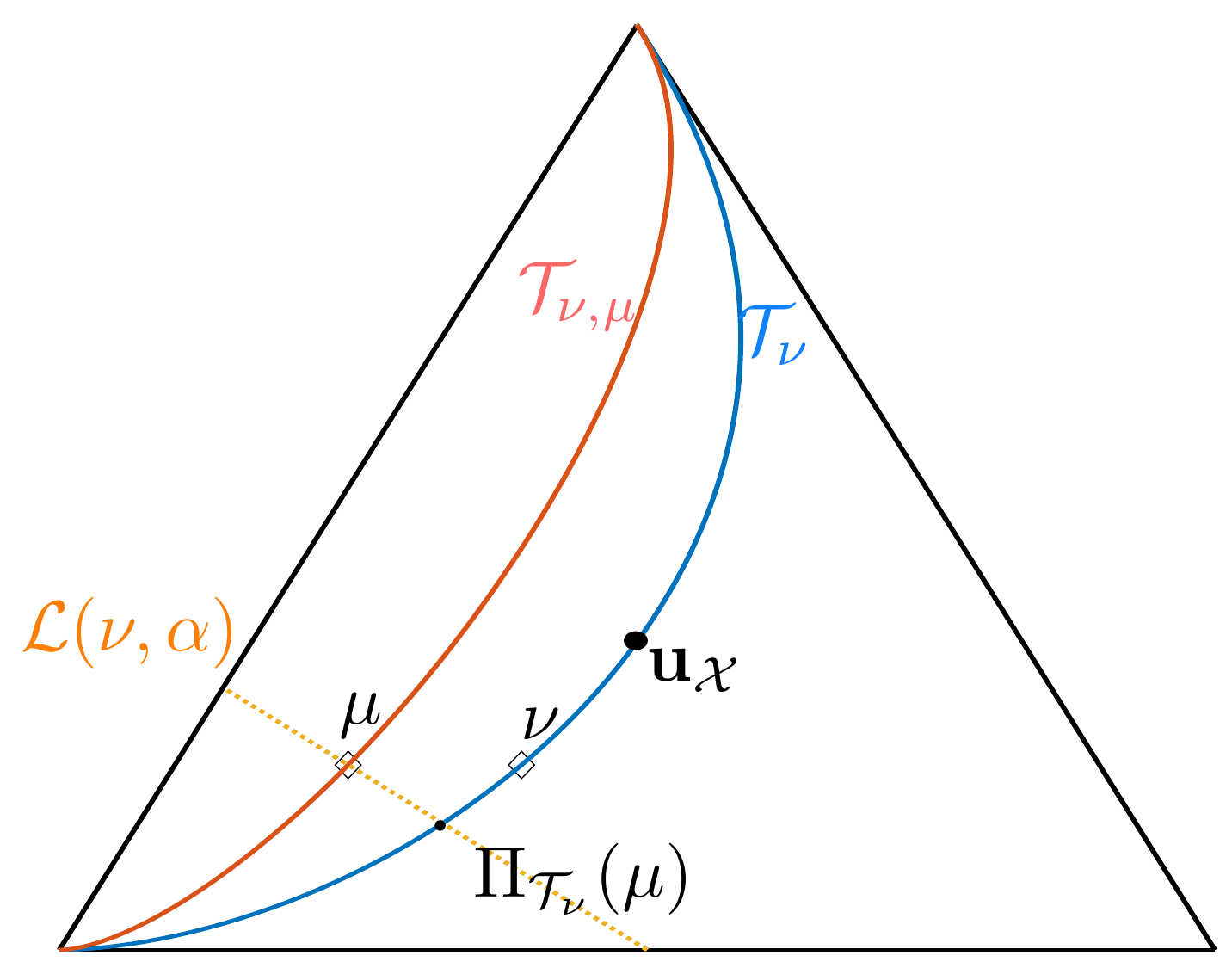}
    \caption{Representation of the 3-dimensional simplex, each point in the triangle represents a distribution over $|\mathcal{X}| =3$. The corners of the triangle correspond to the distribution where all the mass is on a single symbol. The exponential family $\mathcal{T}_\nu$ goes through $\mathbf{u}_\mathcal{X}$ and $\nu$. The exponential family $\mathcal{T}_{\nu,\mu}$ goes through $\mu$. $\mathcal{L}(\nu, \alpha^\star)$ is the linear family of $\nu$ of order $\alpha^\star$ which passes through $\mu$.  The distribution $\Pi_{\mathcal{T}_\nu}(\mu)$ is the projection of $\mu$ onto $\mathcal{T}_\nu$. Of particular interest for lossless coding will be the divergences $D(\mu \| \nu)$ and $D(\mu \| \Pi_{\mathcal{T}_\nu}(\mu))$.}
    \label{fig:main_fig}
\end{figure}

The following lemma guarantees existence and uniqueness of the projection operator.
\begin{lemma}\label{lem:existence_uniqueness}
Let $\mu$ and $\nu$ be umambiguous, then $\Pi_{\mathcal{T}_\nu}(\mu)$ exists and is unique. Further, $\Pi_{\mathcal{T}_\nu}(\mu) = \mu$ iff $\mu \in \mathcal{T}_\nu$.
\end{lemma}
\begin{proof}
Note that for an unambiguous $\nu$, $H(T(\nu,\beta) \| \nu)$ is a strictly decreasing continuous function in $\beta$~\cite{Beirami-IT}. Further, $H(T(\nu, \infty)\| \nu) < H(\mu \| \nu) < H(T(\nu, -\infty)\| \nu)$, thus the projection must exist and is unique, by the intermediate value theorem.
The second part of the claim follows by definition of $\Pi_{\mathcal{T}_\nu}(\mu)$. 
\end{proof}
\noindent The definitions above are summarized in Figure~\ref{fig:main_fig}.
These geometric quantities satisfy various useful properties, which will be of use in the rest of this paper. We will review some of those in the rest of this section. We start with the I-Projection Pythagorean theorem (see for example \cite[Theorem~3.2]{csiszar2004information}).
\begin{lemma}[I-Pythagoerean theorem]\label{thm:pythagore}
Let $\gamma \in \mathcal{T}_\nu$, then
\begin{align}
    D(\mu \| \gamma) = D(\mu \| \Pi_{\mathcal{T}_\nu}(\mu)) + D(\Pi_{\mathcal{T}_\nu}(\mu) \| \gamma).
\end{align}
\end{lemma}

 The next two lemma characterize properties of the projection in terms of entropy and reletive entropy (KL divergence).
\begin{lemma}[Projection does not decrease entropy]\label{lem:entr_ineq}
Let $\Pi_{\mathcal{T}_\nu}(\mu) \in \mathcal{T}_\nu^+$, then
\begin{align}
    H(\Pi_{\mathcal{T}_\nu}(\mu)) &  = H(\mu \| \Pi_{\mathcal{T}_\nu}(\mu)) \geq H(\mu)
\end{align}
with equality iff $\mu \in \mathcal{T}_\nu$.
\end{lemma}
\begin{proof}
We first use the identity $H(\mu) = \log |\mathcal{X}| - D(\mu \| \mathbf{u}_{\mathcal{X}})$. By Theorem~\ref{thm:pythagore}, we have $D(\mu \| \mathbf{u}_\mathcal{X}) = D(\mu \| \Pi_{\mathcal{T}_\nu}(\mu)) + D(\Pi_{\mathcal{T}_\nu}(\mu) \| \mathbf{u}_\mathcal{X})$. Thus, 
\begin{align}
H(\mu) &= \log |\mathcal{X}| - D(\mu \| \Pi_{\mathcal{T}_\nu}(\mu)) - D(\Pi_{\mathcal{T}_\nu}(\mu) \| \mathbf{u}_\mathcal{X}) \\
&\leq \log |\mathcal{X}| - D(\Pi_{\mathcal{T}_\nu}(\mu) \| \mathbf{u}_\mathcal{X}) \\
& = H(\Pi_{\mathcal{T}_\nu}(\mu))
\end{align}
\end{proof}
This yields directly the following lemma.
\begin{lemma}[Projection does not increase relative entropy]\label{cor:div_ineq}
We have
\begin{align}
    D(\Pi_{\mathcal{T}_\nu}(\mu) \| \nu) = D(\mu\|\nu) + H(\mu) - H(\Pi_{\mathcal{T}_\nu}(\mu))\leq D(\mu \| \nu)
\end{align}
with equality iff $\mu \in \mathcal{T}_\nu$ .
\end{lemma}
\begin{proof}
By definition of $\Pi_{\mathcal{T}_\nu}$, we have $H(\Pi_{\mathcal{T}_\nu}(\mu) \| \nu) = H(\mu \| \nu)$, or equivalently that
\begin{align}
    H(\Pi_{\mathcal{T}_\nu}(\mu)) + D(\Pi_{\mathcal{T}_\nu}(\mu) \| \nu) = H(\mu) + D(\mu \| \nu)
\end{align}
The proof follows from using Lemma~\ref{lem:entr_ineq}.
\end{proof}

\section{Matched guesswork}\label{sec:matched}

We define the guesswork $G_\mu(x)$ as the position of $x$ in the list of symbols $\mathcal{X}$ ordered from most likely to least likely according to $\mu$, where ties are broken according to lexicographic ordering. More precisely, let $G_\mu : \mathcal{X} \to |\mathcal{X}|$ be the one-to-one function, such that $G_\mu(x) < G_\mu(y) \implies \mu(x) \geq \mu(y)$, and $x$ is ahead of $y$ in lexicographic ordering.\footnote{Note that the restriction on lexicographic ordering is for convenience and any tie-breaking rule could be used instead with no change to the results in this paper.}
We also consider the logarithm of the guesswork $g_\mu(x) = \log G_\mu(x)$.



We start by reviewing the existing results on guessing with the matched distribution $\mu$ (but not in chronological order). Throughout, we let $X^n \sim \mu^n$, and consider the asymptotic behavior of guesswork as $n \to \infty$.  It was shown in~\cite{ChristiansenDuffy} that, under some mild conditions, the logarithm of guesswork satisfies a  large deviation principle (LDP), and the rate function was further given in terms of information theoretic quantities in~\cite[Theorem 5]{Beirami-IT}.

\begin{theorem}[LDP for matched guesswork]\label{thm:main}
For any unambiguous $\mu$, the sequence $\{\frac{1}{n} g_\mu(X^n) \}_{n \in \mathbb{N}^+}$ satisfies a LDP, with rate function $J(t)$ defined implicitly by
\begin{align}
J(t) = D(T(\mu,\alpha(t))\| \mu),
\end{align}
where $\alpha(t) = \arg_{\alpha \geq 0} \{H(T(\mu,\alpha)) = t\}$.
\end{theorem}
LDP implies many of the results on the average growth rate of the moments, via Varadhan's lemma \cite[Theorem 4.3.1]{Dembo}, which is in essence Laplace's method extended to infinite dimensional spaces. In this setting, one aims at characterizing the normalized growth rate of the $\rho$-th moment of guesswork, denoted $E_\rho(\mu)$, that is for all $\rho> 0$:
\begin{align}
    E_\rho(\mu) \triangleq \frac{1}{\rho} \lim_{n \to \infty} \frac{1}{n} \log \mathbb{E}_{\mu^n}\left[G_{\mu^n}(X^n)^\rho\right].
\end{align}
The following is a direct consequence of Theorem~\ref{thm:main}.
\begin{corollary}\label{lem:guesswork_matched} 
We have,
\begin{align}
E_\rho(\mu) = \max_{\phi \in \mathcal{T}_\mu^+}\left\{ H(\phi) - \frac{1}{\rho}D(\phi\|\mu)\right\} \label{eq:regular_guesswork}.
\end{align}
\end{corollary}
\noindent It is possible to express the solution for the optimization in \eqref{eq:regular_guesswork} in terms of R\'enyi entropies. Indeed, remarking that the optimization \eqref{eq:regular_guesswork} can be equivalently written as an optimization over the tilt parameter, we have that
\begin{align}
E_\rho(\mu) = \max_{\alpha \in \mathbb{R}^+}
\left\{ H(T(\mu,\alpha)) - \frac{1}{\rho}D(T(\mu,\alpha)\|\mu)\right\},
\end{align}
which is maximized by $\alpha = 1/ (1+\rho)$~\cite{Arikan}. This result was originally proved by
Ar{\i}kan, using direct non-asymptotic bounds. He further showed that $E_\rho(\mu) = H_{\frac{1}{1 + \rho}}(\mu)$. In the next section, we follow essentially the same thought process for mismatched guesswork. However, as we shall see, 
there we face further complications that could not be readily resolved using the original techniques used in the proofs of the results for the matched case.

\section{Mismatched Guesswork}\label{sec:mismatched}

In this section, we investigate the behavior of $G_{\nu^n}(X^n)$, when $X^n \sim \mu^n$. For the first time, we establish an LDP for the mismatched guesswork. We also aim at characterizing the exponent of the growth of the moments of mismatched guesswork, denoted by $E_\rho(\nu \| \mu)$, and defined as
\begin{align}
    E_\rho(\nu \| \mu) = \frac{1}{\rho} \lim_{n \to \infty} \mathbb{E}_{\mu^n}\left[G_{\nu^n}(X^n)^\rho \right].
\end{align}
The following result, proved in \cite[Theorem~1]{Beirami-IT}, characterizes the mismatched guesswork in the case where $\mu \in \mathcal{T}_\nu$.
\begin{lemma}[mismatched guesswork on the same tilted family]\label{lem:mism_same_tilted_family}
Let $\mu \in \mathcal{T}_\nu^+$, then $G_\nu(x) = G_\mu(x)$. Alternatively, let $\mu \in \mathcal{T}_\nu^-$, then $G_\nu(x) = |\mathcal{X}| - G_\mu(x)$.
\end{lemma}
\noindent Note that the previous result is non-asymptotic. In particular, it follows readily that $E_\rho(\nu \| \mu) = E_\rho(\mu)$ when $\mu \in \mathcal{T}_\nu^+$ and $E_\rho(\nu \| \mu) = \log(|\mathcal{X}|)$ when $\mu \in \mathcal{T}_\nu^-$. 

However, the techniques in~\cite{Beirami-IT} fall short on characterizing mismatch for $\mu \not \in \mathcal{T}_\nu$. Such characterization is given in the following theorem. 
\begin{theorem} [LDP for mismatched guesswork]\label{thm:mism_ldp}
For any unambiguous $\mu$ and $\nu$, such that $\Pi_{\mathcal{T}_\nu}(\mu) \in \mathcal{T}_\nu^+$, the sequence $\{\frac{1}{n}g_\nu(x^n)\}_{n \in \mathbb{N}^+}$ satisfies a LDP, with rate function $J(t)$, and the rate function is implicitly given by
\begin{align}
J(t) = D(\gamma_{\nu, \mu}(t) \| \mu),
\end{align}
for 
\begin{align}
&\gamma_{\nu, \mu}(t) = \mathcal{T}_{\nu,\mu} \cap \mathcal{L}(\nu,\alpha(t)), \label{eq:gamma_t}\\
&\alpha(t) = \arg_{\alpha\geq 0}\{H(T(\nu,\alpha)) = t\}.\label{eq:alpha_t}
\end{align}
\end{theorem}
Before we proceed to the proof, let us briefly discuss the result. Two features of this result are particularly interesting. First, note that while the value $\alpha(t)$ is determined through a similar implicit equation as the matched guesswork in Theorem~\ref{thm:main}, the rate function is controlled by $D(\gamma_{\nu,\mu}(t) \| \mu)$, where $\gamma_{\nu, \mu}(t) \in \mathcal{T}_{\nu,\mu}$. In particular, if $\mu \in \mathcal{T}_\nu^+$, then Theorem~\ref{thm:mism_ldp} recovers Theorem~\ref{thm:main} by observing that $\nu = T(\mu,\beta)$ for some $\beta > 0$, and thus $\gamma_{\nu,\mu}(t)$ can be reparameterized in terms of $\mu$ only.

The proof of Theorem~\ref{thm:mism_ldp} relies on a correspondence between guesswork, and some sets of distributions, which we will define shortly. This correspondence is implicitly used in \cite[proof of Theorem 5]{Beirami-IT} but it is not explicitly observed. For $\epsilon \geq 0$ and $\alpha \in \mathbb{R}$, let 
\begin{align}
    \mathcal{D}(\nu,\alpha,\epsilon) & \triangleq \left\{ \varphi \in \Delta_{\mathcal{X}}: H(\varphi \| \nu) - H(T(\nu,\alpha)\| \nu) \leq  \epsilon \right\} \\
    \mathcal{E}(\nu,\alpha,\epsilon) & \triangleq \left\{ \varphi \in \Delta_{\mathcal{X}}: H(\varphi \| \nu) - H(T(\nu,\alpha)\| \nu) \geq  - \epsilon \right\} \\
    \mathcal{B}(\nu,\alpha,\epsilon) &\triangleq \left\{\varphi \in \Delta_{\mathcal{X}}:H(T(\nu,\alpha)\|\nu) - H(\varphi \| \nu) \in [0,\epsilon]\right\}.
\end{align}
The sets above are extensions of tilted weakly typical sets of order $\alpha$~\cite[Definition 18]{Beirami-IT}, and capture the set of types which are respectively, more likely, less likely, and as likely according to $\nu$ than $T(\nu,\alpha)$.
For these sets, we then have the following lemma.
\begin{lemma}
For any $\alpha > 0$, the following inclusion relations hold, for sufficiently large $n$,
    \begin{align}
    & \left| \frac{1}{n} g_\nu(x^n) - H(T(\nu,\alpha)) \right| \leq \epsilon  \Rightarrow  \mathbf{q}_{x^n} \in \mathcal{D}(\nu,\alpha,2\epsilon/\alpha), \label{eq:D_set}\\
    & \left| \frac{1}{n} g_\nu(x^n) - H(T(\nu,\alpha)) \right| \leq \epsilon  \Rightarrow  \mathbf{q}_{x^n} \in \mathcal{E}(\nu,\alpha,2\epsilon/\alpha), \label{eq:E_set}\\
    &\left| \frac{1}{n} g_\nu(x^n) - H(T(\nu,\alpha)) \right| \leq \epsilon \Leftarrow \mathbf{q}_{x^n} \in \mathcal{B}(\nu,\alpha,\epsilon/\alpha).\label{eq:B_set}
\end{align}
\end{lemma}
This was proved implicitly in the proofs of Theorems 3 and 5 in~\cite{Beirami-IT}.
We are now equipped to provide the proof of the main theorem.

\definecolor{mycolor1}{rgb}{0.00000,0.44700,0.74100}%
\definecolor{mycolor2}{rgb}{0.85000,0.32500,0.09800}%
\definecolor{mycolor3}{rgb}{0.00000,0.49804,0.00000}%

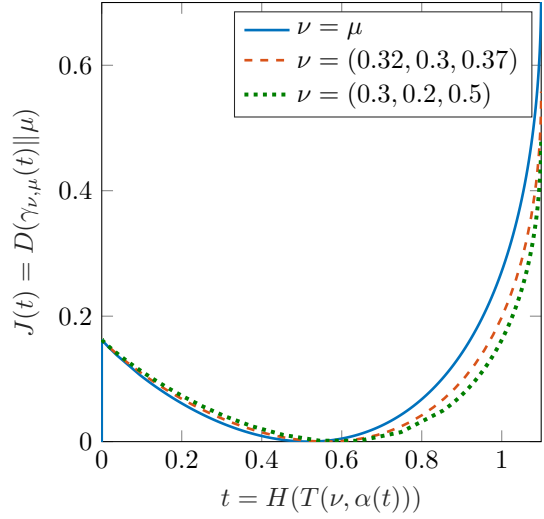
\begin{figure}
\centering
\begin{tikzpicture}

\begin{axis}[%
width=2.3in,
height=2.3in,
at={(1.273in,0.674in)},
scale only axis,
xmin=0,
xmax=1.0986,
xlabel style={font=\color{white!15!black}},
xlabel={$t = H(T(\nu,\alpha(t)))$},
ymin=0,
ymax=0.7,
ylabel style={font=\color{white!15!black}},
ylabel={$J(t) = D(\gamma_{\nu,\mu}(t) \| \mu)$},
axis background/.style={fill=white},
legend style={legend cell align=left, align=left, draw=white!15!black}
]

\addplot [color=mycolor1, line width=1.0pt]
  table[row sep=crcr]{%
5.69810126605036e-36	-0.165396390147414\\
1.64050187223172e-35	-0.165396390147414\\
4.72230142424492e-35	-0.165396390147414\\
1.35912465237657e-34	-0.165396390147414\\
3.91103425649944e-34	-0.165396390147414\\
1.12524930635456e-33	-0.165396390147414\\
3.23689545703947e-33	-0.165396390147414\\
9.30956259463682e-33	-0.165396390147414\\
2.67700002665843e-32	-0.165396390147414\\
7.69633017252035e-32	-0.165396390147414\\
2.2122430776284e-31	-0.165396390147414\\
6.35760231874108e-31	-0.165396390147414\\
1.82668071274069e-30	-0.165396390147414\\
5.24732539234831e-30	-0.165396390147414\\
1.5070115881272e-29	-0.165396390147414\\
4.32708549192297e-29	-0.165396390147414\\
1.24214288495083e-28	-0.165396390147414\\
3.56485274066123e-28	-0.165396390147414\\
1.02282703214198e-27	-0.165396390147414\\
2.93393039180729e-27	-0.165396390147414\\
8.41357720506441e-27	-0.165396390147414\\
2.41207556192484e-26	-0.165396390147414\\
6.91315594535163e-26	-0.165396390147414\\
1.98076374879718e-25	-0.165396390147414\\
5.67355558293439e-25	-0.165396390147414\\
1.62457378209124e-24	-0.165396390147414\\
4.6502893310898e-24	-0.165396390147414\\
1.33067362533946e-23	-0.165396390147414\\
3.80634797151097e-23	-0.165396390147414\\
1.08839061251199e-22	-0.165396390147414\\
3.11095753808713e-22	-0.165396390147414\\
8.88852358349466e-22	-0.165396390147414\\
2.53854145123728e-21	-0.165396390147414\\
7.24686810151469e-21	-0.165396390147414\\
2.06785372032806e-20	-0.165396390147414\\
5.89771778332824e-20	-0.165396390147414\\
1.68125524664591e-19	-0.165396390147414\\
4.79025851315772e-19	-0.165396390147414\\
1.36410980114155e-18	-0.165396390147414\\
3.88233915587993e-18	-0.165396390147414\\
1.10428005993813e-17	-0.165396390147414\\
3.13901604672604e-17	-0.165396390147414\\
8.91707274231496e-17	-0.165396390147414\\
2.53133925488241e-16	0.162518929497775\\
7.18060828746545e-16	0.162518929497774\\
2.03534085834697e-15	0.162518929497773\\
5.98650648066045e-15	0.16251892949777\\
1.67560088746288e-14	0.162518929497759\\
4.74484284215328e-14	0.16251892949773\\
1.3424687037905e-13	0.162518929497649\\
3.78930876930705e-13	0.162518929497421\\
1.06914541782032e-12	0.162518929496777\\
3.0129492778173e-12	0.16251892949497\\
8.48077607055745e-12	0.162518929489902\\
2.38413643921641e-11	0.162518929475704\\
6.69325588911817e-11	0.162518929436004\\
1.87633680042455e-10	0.162518929325192\\
5.25171710385576e-10	0.162518929016487\\
1.46742906398864e-09	0.162518928158314\\
4.09279440020968e-09	0.162518925778166\\
1.13926979743423e-08	0.162518919193503\\
3.16452184981101e-08	0.162518901027535\\
8.76980351580497e-08	0.16251885106331\\
2.42433075979925e-07	0.162518714101266\\
6.68389215484781e-07	0.162518340051977\\
1.83744545872935e-06	0.162517322700415\\
5.03567243980062e-06	0.162514568296482\\
1.37553127065777e-05	0.162507148926443\\
3.74429728433717e-05	0.162487278204754\\
0.000101550474397656	0.162434406186378\\
0.000274371279723133	0.162294784172581\\
0.000738368696953217	0.161929396285093\\
0.00197872606270698	0.160984004722626\\
0.00527777881816528	0.158575460336145\\
0.0139921711423854	0.152580749142067\\
0.0367403525515901	0.138246301094247\\
0.0946597322741494	0.106612523747575\\
0.0964377016276169	0.105725775568152\\
0.098247441748898	0.104827752017456\\
0.100089459163482	0.103918410591087\\
0.101964266058076	0.10299771308882\\
0.103872380239804	0.102065625838437\\
0.105814325089531	0.10112211992711\\
0.107790629509033	0.100167171440539\\
0.109801827861749	0.0992007617099997\\
0.111848459906831	0.0982228775675145\\
0.113931070726209	0.0972335116092936\\
0.116050210644364	0.0962326624676537\\
0.11820643514051	0.0952203350915823\\
0.120400304752875	0.0941965410361321\\
0.12263238497475	0.0931612987608238\\
0.124903246141972	0.0921146339372354\\
0.12721346331152	0.0910565797659548\\
0.129563616130855	0.0899871773030672\\
0.131954288697657	0.0889064757963504\\
0.13438606940959	0.0878145330313442\\
0.136859550803732	0.0867114156874578\\
0.139375329385264	0.0855971997042738\\
0.141934005445056	0.084471970658204\\
0.144536182865733	0.0833358241496433\\
0.147182468915813	0.0821888662007653\\
0.149873474031519	0.0810312136640928\\
0.152609811585825	0.0798629946419705\\
0.155392097644317	0.0786843489170561\\
0.158220950707431	0.0774954283939402\\
0.161096991438623	0.0762963975519896\\
0.164020842378021	0.0750874339095013\\
0.166993127641112	0.0738687284992401\\
0.170014472601996	0.0726404863554179\\
0.17308550356074	0.0714029270121612\\
0.176206847394374	0.0701562850134933\\
0.179379131191044	0.0689008104348436\\
0.182602981866855	0.0676367694160766\\
0.185879025764928	0.0663644447060127\\
0.189207888236176	0.0650841362183937\\
0.192590193201345	0.06379616159922\\
0.19602656269382	0.0625008568053661\\
0.199517616382735	0.0611985766943512\\
0.203063971075901	0.0598896956251178\\
0.206666240202091	0.0585746080696389\\
0.210325033272215	0.057253729235146\\
0.214040955318929	0.0559274956967363\\
0.217814606314226	0.0545963660400822\\
0.221646580564577	0.053260821513931\\
0.225537466083179	0.0519213666920431\\
0.229487843938925	0.0505785301441779\\
0.233498287581662	0.0492328651156922\\
0.237569362143387	0.0478849502152715\\
0.241701623714997	0.04653539011027\\
0.245895618598266	0.0451848162290825\\
0.250151882532733	0.0438338874699221\\
0.254470939897198	0.0424832909153243\\
0.258853302885587	0.0411337425516404\\
0.263299470656939	0.0397859879927267\\
0.267809928459329	0.0384408032069739\\
0.272385146727579	0.0370989952467595\\
0.277025580154622	0.0357614029793396\\
0.281731666736467	0.0344288978181299\\
0.286503826790736	0.0331023844532554\\
0.291342461948788	0.0317828015801757\\
0.296247954121532	0.0304711226251196\\
0.301220664439044	0.0291683564659858\\
0.306260932164214	0.0278755481472873\\
0.311369073580663	0.026593779587638\\
0.316545380855296	0.0253241702781958\\
0.321790120875865	0.0240678779703946\\
0.327103534064064	0.0228260993512098\\
0.332485833164703	0.0216000707041162\\
0.337937202011622	0.0203910685538079\\
0.343457794271093	0.0192004102926608\\
0.349047732163558	0.0180294547868274\\
0.354707105164621	0.0168796029597635\\
0.360435968686372	0.0157522983508943\\
0.366234342740156	0.0146490276470365\\
0.37210221058209	0.0135713211841018\\
0.378039517342674	0.0125207534165177\\
0.384046168642027	0.0114989433517106\\
0.390122029192354	0.0105075549469094\\
0.396266921389409	0.00954829746544272\\
0.402480623894834	0.00862292578961713\\
0.408762870211409	0.00773324068718378\\
0.415113347253362	0.00688108902832491\\
0.421531693914054	0.0060683639500161\\
0.428017499633483	0.00529700496455208\\
0.434570302968203	0.00456899800896066\\
0.441189590166415	0.0038863754319701\\
0.447874793751111	0.00325121591514495\\
0.454625291114336	0.00266564432475974\\
0.461440403125761	0.00213183149094484\\
0.468319392758925	0.00165199391060929\\
0.47526146373865	0.00122839337062869\\
0.482265759213301	0.000863336487776446\\
0.48933136045567	0.000559174161880766\\
0.496457285596477	0.000318300938703319\\
0.503642488394543	0.000143154279063618\\
0.510885857047902	3.62137307733134e-05\\
0.518186213050213	1.01013113691836e-35\\
0.525542310096954	3.70739187502014e-05\\
0.532952833046039	0.000150035305247391\\
0.540416396937568	0.000341521714083809\\
0.547931546077554	0.000614207073145191\\
0.555496753190544	0.000970800204446317\\
0.563110418646155	0.00141404322617272\\
0.570770869764584	0.00194670983340133\\
0.578476360206241	0.0025716034551702\\
0.586225069450666	0.00329155528578553\\
0.594015102369936	0.00410942218849391\\
0.601844488901762	0.00502808446990842\\
0.609711183827479	0.00605044352386051\\
0.617613066660084	0.00717941934365431\\
0.625547941647445	0.0084179479020274\\
0.633513537895717	0.00976897839847151\\
0.641507509617914	0.0112354703739383\\
0.649527436512462	0.0128203906933476\\
0.657570824276422	0.0145267103967317\\
0.665635105257884	0.016357401420282\\
0.673717639251859	0.0183154331890227\\
0.681815714443765	0.0204037690833062\\
0.689926548504334	0.0226253627818218\\
0.698047289839521	0.0249831544843143\\
0.706175018998661	0.0274800670177348\\
0.714306750243791	0.0301190018300839\\
0.722439433282698	0.0329028348767566\\
0.730569955167844	0.0358344124047565\\
0.738695142362916	0.0389165466407167\\
0.746811762978278	0.0421520113892319\\
0.754916529176149	0.0455435375485865\\
0.763006099745799	0.049093808551533\\
0.771077082848571	0.0528054557393486\\
0.779126038931924	0.0566810536779655\\
0.787149483811196	0.0607231154255192\\
0.795143891917111	0.0649340877612126\\
0.803105699706519	0.0693163463859124\\
0.811031309233172	0.0738721911054051\\
0.818917091874739	0.0786038410077269\\
0.826759392211585	0.083513429646435\\
0.834554532052201	0.0886030002421185\\
0.842298814599502	0.0938745009148377\\
0.849988528751534	0.0993297799605386\\
0.857619953529495	0.1049705811848\\
0.865189362625293	0.110798539307543\\
0.872693029060231	0.11681517545255\\
0.880127229945785	0.123021892735813\\
0.887488251336785	0.12941997196685\\
0.894772393166783	0.13601056747717\\
0.901975974254745	0.142794703090091\\
0.909095337371721	0.149773268246027\\
0.91612685435562	0.156947014297257\\
0.923066931261731	0.164316550985971\\
0.929912013536245	0.171882343119154\\
0.936658591199583	0.179644707453522\\
0.943303204026073	0.187603809803351\\
0.949842446706166	0.195759662383556\\
0.956272973977198	0.204112121399876\\
0.962591505708498	0.21266088489739\\
0.968794831926544	0.221405490877974\\
0.974879817765819	0.230345315696554\\
0.980843408330998	0.239479572745234\\
0.986682633456215	0.248807311433565\\
0.992394612347267	0.258327416472285\\
0.997976558092817	0.268038607466955\\
1.00342578203095	0.27793943882691\\
1.00873969795779	0.288028299993909\\
1.01391582616516	0.298303415993824\\
1.01895179729497	0.308762848313578\\
1.02384535599837	0.319404496104455\\
1.0285943643882	0.330226097711733\\
1.03319680527428	0.341225232529456\\
1.03765078517134	0.352399323177995\\
1.04195453707055	0.363745638000874\\
1.04610642296598	0.375261293876201\\
1.05010493612862	0.386943259336892\\
1.05394870312107	0.398788357992739\\
1.05763648554721	0.410793272246298\\
1.06116718153185	0.422954547293498\\
1.06453982692664	0.435268595398821\\
1.06775359623915	0.447731700433969\\
1.07080780328324	0.460340022667945\\
1.07370190154985	0.47308960379562\\
1.07643548429824	0.485976372191051\\
1.07900828436873	0.49899614837103\\
1.08142017371903	0.51214465065366\\
1.08367116268719	0.525417500996159\\
1.08576139898514	0.538810230995529\\
1.08769116642768	0.552318288035253\\
1.08946088340283	0.565937041560832\\
1.09107110109007	0.579661789466612\\
1.09252250143397	0.593487764576184\\
1.09381589488149	0.607410141198445\\
1.09495221789192	0.621424041741377\\
1.09593253022903	0.63552454336559\\
1.09675801204577	0.649706684659802\\
1.09742996077224	0.663965472320568\\
1.09794978781844	0.678295887818836\\
1.09831901510338	0.692692894036222\\
1.09853927142282	0.707151441854279\\
};
\addlegendentry{$\nu=\mu$}

\addplot [color=mycolor2, dashed, line width=1.0pt]
  table[row sep=crcr]{%
0.0225272768963477	0.148396188538835\\
0.0240325971612508	0.147525832136216\\
0.0256373111722396	0.146604324489889\\
0.0273478520712815	0.145628926332639\\
0.0291710525062492	0.144596783872479\\
0.0311141669810326	0.143504928424139\\
0.0331848950722413	0.142350276823065\\
0.0353914054802113	0.141129632764549\\
0.0377423608671938	0.139839689230105\\
0.0402469434180846	0.138477032184422\\
0.0429148810384713	0.137038145750238\\
0.0457564740807024	0.135519419094748\\
0.048782622460653	0.133917155290319\\
0.0520048529953558	0.132227582444381\\
0.0554353467541086	0.130446867428506\\
0.0590869661724531	0.128571132575073\\
0.0629732816288694	0.126596475751777\\
0.0671085971274291	0.124518994269183\\
0.0715079746652615	0.12233481312536\\
0.0761872567907123	0.120040118143166\\
0.0811630867757189	0.117631194610714\\
0.0864529257334194	0.115104472093123\\
0.0920750659085512	0.112456576142912\\
0.0980486392531492	0.109684387697377\\
0.104393620272797	0.10678511101183\\
0.111130821988863	0.103756351036809\\
0.118281883709644	0.10059620120322\\
0.125869249138307	0.0973033426286626\\
0.133916133168784	0.0938771557988421\\
0.142446475533527	0.0903178458048508\\
0.151484879271507	0.0866265822260957\\
0.161056531784112	0.0828056547336748\\
0.171187106045041	0.0788586454430252\\
0.181902639333781	0.0747906189594724\\
0.193229386678501	0.0706083309264383\\
0.205193646032992	0.0663204556924611\\
0.217821552085959	0.0619378334474346\\
0.231138835524633	0.0574737368272136\\
0.24517054456656	0.0529441565338769\\
0.25994072565521	0.0483681049511588\\
0.275472060411776	0.0437679360352374\\
0.291785456275958	0.0391696789148942\\
0.308899588784415	0.0346033816289164\\
0.326830394161394	0.0301034602521966\\
0.345590511867788	0.0257090473104009\\
0.365188678008511	0.0214643318586361\\
0.385629072066906	0.0174188819149546\\
0.40691062134713	0.0136279381208081\\
0.429026269779252	0.0101526655906048\\
0.451962220381255	0.00706034897478002\\
0.475697163660836	0.00442451388245349\\
0.500201507534859	0.00232495610459493\\
0.525436627868865	0.000847658688668241\\
0.551354162376649	8.45760105273438e-05\\
0.577895374208306	0.00013326376313059\\
0.604990614883185	0.00109633444651557\\
0.632558919031935	0.00308071971815321\\
0.660507765395127	0.00619672405662784\\
0.688733039347971	0.0105568587831302\\
0.71711923153006	0.0162744516963272\\
0.745539904617204	0.0234620354424588\\
0.773858455587107	0.0322295271831252\\
0.801929193797991	0.0426822229126425\\
0.829598745747151	0.0549186415255047\\
0.85670778560796	0.0690282658746918\\
0.883093076872743	0.0850892398684895\\
0.90858979520077	0.103166091255349\\
0.933034086662374	0.123307558201865\\
0.956265799962464	0.145544603124306\\
0.978131317045779	0.169888698631176\\
0.998486394915895	0.196330467208778\\
1.01719892367089	0.224838748068007\\
1.03415150262372	0.255360151382478\\
1.04924373858969	0.287819142430407\\
1.06239417824288	0.322118676784783\\
1.07354179966077	0.358141383944349\\
1.08264700610986	0.395751272215205\\
1.08280811904471	0.396518637460999\\
1.08296840714838	0.397286575786182\\
1.08312787037234	0.398055085906931\\
1.08328650867159	0.398824166537211\\
1.0834443220046	0.399593816388794\\
1.08360131033335	0.400364034171245\\
1.08375747362332	0.401134818591957\\
1.08391281184348	0.401906168356147\\
1.08406732496629	0.402678082166863\\
1.08422101296773	0.403450558724997\\
1.08437387582722	0.404223596729296\\
1.08452591352771	0.404997194876368\\
1.08467712605563	0.405771351860685\\
1.08482751340089	0.406546066374609\\
1.08497707555688	0.40732133710838\\
1.08512581252048	0.408097162750148\\
1.08527372429206	0.408873541985958\\
1.08542081087545	0.409650473499775\\
1.08556707227798	0.410427955973495\\
1.08571250851044	0.411205988086938\\
1.0858571195871	0.41198456851787\\
1.08600090552569	0.412763695942016\\
1.08614386634745	0.413543369033061\\
1.08628600207704	0.414323586462656\\
1.08642731274263	0.415104346900433\\
1.08656779837581	0.415885649014018\\
1.08670745901168	0.416667491469034\\
1.08684629468876	0.417449872929109\\
1.08698430544906	0.418232792055897\\
1.08712149133804	0.419016247509063\\
1.0872578524046	0.419800237946329\\
1.08739338870111	0.420584762023445\\
1.08752810028338	0.421369818394221\\
1.08766198721067	0.422155405710539\\
1.08779504954569	0.42294152262234\\
1.08792728735459	0.423728167777663\\
1.08805870070696	0.424515339822635\\
1.08818928967585	0.42530303740148\\
1.08831905433771	0.426091259156534\\
1.08844799477246	0.426880003728258\\
1.08857611106342	0.427669269755247\\
1.08870340329736	0.428459055874228\\
1.08882987156448	0.429249360720078\\
1.08895551595839	0.430040182925841\\
1.08908033657613	0.430831521122723\\
1.08920433351815	0.431623373940108\\
1.08932750688834	0.432415740005569\\
1.08944985679398	0.433208617944887\\
1.08957138334577	0.434002006382027\\
1.08969208665781	0.434795903939195\\
1.08981196684763	0.435590309236804\\
1.08993102403613	0.436385220893516\\
1.09004925834763	0.437180637526238\\
1.09016666990986	0.43797655775012\\
1.09028325885391	0.438772980178591\\
1.09039902531429	0.439569903423353\\
1.09051396942887	0.440367326094386\\
1.09062809133895	0.441165246799962\\
1.09074139118917	0.441963664146672\\
1.09085386912756	0.442762576739405\\
1.09096552530554	0.443561983181384\\
1.0910763598779	0.444361882074158\\
1.09118637300278	0.445162272017624\\
1.09129556484172	0.445963151610035\\
1.09140393555958	0.446764519447998\\
1.09151148532462	0.447566374126492\\
1.09161821430845	0.448368714238896\\
1.09172412268601	0.44917153837696\\
1.09182921063561	0.449974845130849\\
1.09193347833891	0.450778633089133\\
1.0920369259809	0.451582900838812\\
1.09213955374992	0.452387646965309\\
1.09224136183764	0.4531928700525\\
1.09234235043908	0.453998568682697\\
1.09244251975256	0.454804741436689\\
1.09254186997975	0.455611386893731\\
1.09264040132564	0.456418503631557\\
1.09273811399853	0.457226090226396\\
1.09283500821003	0.458034145252974\\
1.09293108417509	0.458842667284536\\
1.09302634211194	0.459651654892841\\
1.09312078224212	0.460461106648184\\
1.09321440479049	0.461271021119402\\
1.09330720998516	0.46208139687388\\
1.0933991980576	0.462892232477571\\
1.0934903692425	0.463703526494989\\
1.09358072377789	0.46451527748924\\
1.09367026190504	0.465327484022018\\
1.09375898386853	0.466140144653619\\
1.09384688991619	0.466953257942951\\
1.09393398029914	0.467766822447539\\
1.09402025527174	0.468580836723549\\
1.09410571509162	0.469395299325784\\
1.09419036001969	0.470210208807701\\
1.09427419032008	0.471025563721418\\
1.09435720626019	0.471841362617725\\
1.09443940811065	0.472657604046096\\
1.09452079614535	0.473474286554703\\
1.09460137064138	0.474291408690406\\
1.09468113187911	0.475108968998796\\
1.09476008014211	0.475926966024176\\
1.09483821571716	0.476745398309586\\
1.09491553889429	0.477564264396804\\
1.09499204996672	0.478383562826375\\
1.09506774923089	0.479203292137595\\
1.09514263698646	0.480023450868538\\
1.09521671353626	0.480844037556064\\
1.09528997918634	0.481665050735825\\
1.09536243424594	0.482486488942284\\
1.09543407902747	0.483308350708713\\
1.09550491384654	0.48413063456721\\
1.09557493902194	0.484953339048707\\
1.09564415487563	0.485776462682984\\
1.09571256173272	0.486600003998677\\
1.09578015992152	0.487423961523287\\
1.09584694977347	0.488248333783195\\
1.09591293162319	0.489073119303657\\
1.09597810580842	0.489898316608842\\
1.09604247267006	0.490723924221817\\
1.09610603255217	0.491549940664567\\
1.09616878580191	0.492376364458\\
1.0962307327696	0.493203194121974\\
1.09629187380868	0.494030428175287\\
1.09635220927568	0.494858065135687\\
1.0964117395303	0.495686103519915\\
1.09647046493531	0.496514541843667\\
1.0965283858566	0.497343378621639\\
1.09658550266315	0.498172612367522\\
1.09664181572706	0.499002241594027\\
1.09669732542351	0.499832264812868\\
1.09675203213074	0.500662680534807\\
1.09680593623011	0.501493487269637\\
1.09685903810604	0.502324683526202\\
1.09691133814601	0.503156267812408\\
1.09696283674058	0.503988238635238\\
1.09701353428335	0.50482059450075\\
1.097063431171	0.505653333914098\\
1.09711252780325	0.506486455379536\\
1.09716082458285	0.50731995740043\\
1.0972083219156	0.508153838479274\\
1.09725502021034	0.508988097117691\\
1.09730091987893	0.509822731816455\\
1.09734602133625	0.510657741075481\\
1.09739032500019	0.511493123393857\\
1.09743383129167	0.512328877269848\\
1.09747654063461	0.513165001200902\\
1.09751845345592	0.514001493683651\\
1.09755957018551	0.514838353213956\\
1.0975998912563	0.515675578286871\\
1.09763941710415	0.516513167396685\\
1.09767814816793	0.517351119036929\\
1.09771608488949	0.518189431700376\\
1.09775322771362	0.519028103879052\\
1.09778957708808	0.51986713406426\\
1.0978251334636	0.520706520746575\\
1.09785989729384	0.521546262415857\\
1.09789386903542	0.522386357561272\\
1.09792704914788	0.523226804671291\\
1.09795943809371	0.5240676022337\\
1.09799103633832	0.52490874873562\\
1.09802184435004	0.525750242663514\\
1.09805186260011	0.526592082503188\\
1.09808109156269	0.527434266739813\\
1.09810953171483	0.528276793857924\\
1.09813718353649	0.529119662341453\\
1.0981640475105	0.529962870673698\\
1.09819012412261	0.530806417337383\\
1.09821541386142	0.531650300814636\\
1.09823991721841	0.53249451958699\\
1.09826363468793	0.533339072135434\\
1.0982865667672	0.534183956940394\\
1.09830871395626	0.53502917248174\\
1.09833007675805	0.535874717238813\\
1.0983506556783	0.536720589690418\\
1.09837045122563	0.537566788314858\\
1.09838946391144	0.538413311589922\\
1.09840769424999	0.5392601579929\\
1.09842514275833	0.540107326000603\\
1.09844180995635	0.540954814089364\\
1.09845769636672	0.541802620735048\\
1.09847280251493	0.542650744413074\\
1.09848712892925	0.543499183598402\\
1.09850067614074	0.54434793676558\\
1.09851344468322	0.545197002388706\\
1.09852543509333	0.546046378941479\\
1.09853664791044	0.546896064897193\\
1.09854708367668	0.547746058728751\\
1.09855674293696	0.548596358908668\\
1.09856562623891	0.54944696390908\\
1.09857373413293	0.55029787220178\\
1.09858106717212	0.551149082258185\\
1.09858762591235	0.552000592549391\\
1.09859341091217	0.55285240154614\\
1.09859842273287	0.553704507718872\\
1.09860266193846	0.554556909537704\\
1.09860612909561	0.555409605472451\\
1.09860882477373	0.556262593992645\\
1.0986107495449	0.557115873567531\\
1.09861190398388	0.55796944266608\\
};
\addlegendentry{$\nu =(0.32,0.3,0.37)$}

\addplot [color=mycolor3, dotted, line width=1.5pt]
  table[row sep=crcr]{%
2.86505889391656e-08	0.162518904634483\\
3.65469773615981e-08	0.162518897840173\\
4.66129085617449e-08	0.162518889197021\\
5.94423839635973e-08	0.162518878204349\\
7.57913953468532e-08	0.162518864226639\\
9.66219024536649e-08	0.162518846457463\\
1.23157689122754e-07	0.162518823873842\\
1.56955285004834e-07	0.162518795178465\\
1.99993992625764e-07	0.162518758726686\\
2.54790134194962e-07	0.162518712434233\\
3.245419876383e-07	0.162518653660575\\
4.13313603000285e-07	0.162518579061742\\
5.26267848642994e-07	0.162518484404447\\
6.69961779510938e-07	0.162518364331741\\
8.52720861168807e-07	0.162518212067468\\
1.0851129320577e-06	0.162518019043917\\
1.38054826598924e-06	0.162517774432967\\
1.75603898705403e-06	0.162517464555984\\
2.23315977172377e-06	0.162517072141767\\
2.83926268127869e-06	0.162516575393927\\
3.60901267982223e-06	0.162515946819701\\
4.5863276145059e-06	0.162515151760304\\
5.82682803468682e-06	0.162514146548232\\
7.40092931790669e-06	0.162512876198896\\
9.39774250081341e-06	0.162511271521443\\
1.19299926799878e-05	0.162509245506278\\
1.51402169519669e-05	0.162506688812715\\
1.920857019082e-05	0.162503464138875\\
2.43626497332818e-05	0.162499399205016\\
3.08898532200609e-05	0.162494278019693\\
3.91529122884104e-05	0.162487830022831\\
4.96094045504682e-05	0.162479716608536\\
6.28362446577653e-05	0.162469514420377\\
7.95604012870764e-05	0.162456694679397\\
0.000100697391564577	0.16244059764644\\
0.000127399481234864	0.162420401131316\\
0.000161115984079002	0.162395081737096\\
0.000203668627513363	0.162363367263401\\
0.000257345656942318	0.162323678382632\\
0.000325019218227363	0.162274057342112\\
0.000410291620471864	0.162212081027801\\
0.000517677382106912	0.162134755246082\\
0.000652829552274251	0.1620383865343\\
0.000822820736955248	0.161918427194183\\
0.00103649161711281	0.161769288551142\\
0.00130488261023672	0.161584116674869\\
0.00164176779960196	0.161354523952199\\
0.00206431445943607	0.161070268985133\\
0.00259389658453962	0.160718876305935\\
0.00325709696260916	0.160285186382945\\
0.00408693971104925	0.159750825389894\\
0.00512440407697848	0.159093583339127\\
0.00642028094459525	0.158286688644732\\
0.00803744621228526	0.157297967363242\\
0.0100536403033489	0.156088876908105\\
0.0125648608198707	0.154613408027743\\
0.0156894958285231	0.152816857014234\\
0.0195733481502406	0.150634485225015\\
0.0243957250813337	0.147990109395243\\
0.0303767901582733	0.144794710750472\\
0.0377863873641305	0.140945224479354\\
0.0469545405508149	0.136323790989146\\
0.0582837770863947	0.130797943880449\\
0.0722632794989918	0.124222519297364\\
0.089484551776737	0.116444561647057\\
0.110657660274926	0.107313264616277\\
0.136625946778179	0.0966981473602692\\
0.16837506326005	0.0845203614137885\\
0.207028751452112	0.0708043413864201\\
0.253818407673766	0.0557598114396358\\
0.310005751367734	0.0399066711813004\\
0.376728454556882	0.0242553390150503\\
0.454730703754939	0.0105479618997305\\
0.543943336830917	0.00154305121085232\\
0.642909159744182	0.00127661788548695\\
0.748132826165357	0.0151519186412888\\
0.853583679141897	0.0496211371803508\\
0.855640061796641	0.0505667816882182\\
0.857693106255408	0.0515233751297522\\
0.859742734969961	0.0524909668536775\\
0.861788870137103	0.0534696059472123\\
0.863831433704067	0.0544593412293389\\
0.865870347373982	0.055460221244049\\
0.867905532611434	0.056472294253563\\
0.869936910648089	0.0574956082315238\\
0.871964402488417	0.0585302108561666\\
0.873987928915474	0.0595761495034681\\
0.876007410496784	0.0606334712402714\\
0.878022767590289	0.0617022228173917\\
0.880033920350379	0.0627824506627027\\
0.882040788734004	0.0638742008742048\\
0.884043292506864	0.0649775192130756\\
0.886041351249673	0.0660924510967048\\
0.888034884364507	0.0672190415917141\\
0.890023811081221	0.0683573354069642\\
0.892008050463948	0.0695073768865485\\
0.893987521417677	0.0706692100027772\\
0.895962142694894	0.0718428783491491\\
0.89793183290231	0.0730284251333185\\
0.899896510507662	0.0742258931700506\\
0.901856093846575	0.0754353248741741\\
0.903810501129516	0.0766567622535287\\
0.905759650448803	0.0778902469019076\\
0.907703459785699	0.0791358199919999\\
0.909641847017563	0.0803935222683312\\
0.911574729925085	0.0816633940402067\\
0.913502026199579	0.0829454751746544\\
0.915423653450353	0.0842398050893735\\
0.917339529212139	0.0855464227456866\\
0.919249570952591	0.0868653666414995\\
0.921153696079856	0.0881966748042697\\
0.923051821950203	0.0895403847839813\\
0.924943865875712	0.0908965336461345\\
0.926829745132039	0.0922651579647445\\
0.928709376966227	0.093646293815356\\
0.930582678604593	0.0950399767680737\\
0.932449567260662	0.096446241880607\\
0.934309960143169	0.0978651236913377\\
0.936163774464111	0.099296656212402\\
0.938010927446864	0.1007408729228\\
0.939851336334346	0.102197806761524\\
0.941684918397241	0.103667490120714\\
0.943511590942274	0.10514995483884\\
0.945331271320535	0.106645232193909\\
0.947143876935858	0.108153352896706\\
0.948949325253247	0.109674347084061\\
0.950747533807345	0.111208244312157\\
0.952538420210961	0.112755073549862\\
0.954321902163629	0.114314863172105\\
0.956097897460221	0.115887640953289\\
0.957866323999596	0.117473434060739\\
0.959627099793295	0.119072269048194\\
0.961380142974267	0.120684171849343\\
0.963125371805645	0.122309167771402\\
0.964862704689549	0.123947281488742\\
0.966592060175926	0.125598537036558\\
0.968313356971427	0.127262957804594\\
0.970026513948311	0.128940566530914\\
0.971731450153384	0.130631385295731\\
0.97342808481696	0.132335435515282\\
0.975116337361859	0.134052737935771\\
0.97679612741242	0.135783312627357\\
0.978467374803543	0.137527178978214\\
0.980129999589754	0.139284355688638\\
0.981783922054287	0.141054860765231\\
0.983429062718187	0.142838711515142\\
0.985065342349429	0.14463592454037\\
0.98669268197205	0.146446515732147\\
0.988311002875298	0.148270500265378\\
0.989920226622791	0.150107892593159\\
0.991520275061686	0.151958706441366\\
0.99311107033185	0.153822954803318\\
0.99469253487505	0.155700649934516\\
0.996264591444133	0.157591803347463\\
0.997827163112218	0.159496425806553\\
0.999380173281882	0.161414527323055\\
1.00092354569435	0.163346117150168\\
1.00245720443867	0.165291203778165\\
1.00398107396092	0.167249794929621\\
1.00549507907333	0.169221897554733\\
1.00699914496348	0.171207517826716\\
1.00849319720345	0.173206661137307\\
1.00997716175895	0.175219332092342\\
1.01145096499842	0.177245534507442\\
1.01291453370218	0.179285271403785\\
1.01436779507146	0.181338545003975\\
1.01581067673752	0.183405356728011\\
1.01724310677065	0.185485707189351\\
1.0186650136892	0.187579596191084\\
1.02007632646857	0.189687022722194\\
1.02147697455018	0.191807984953937\\
1.02286688785036	0.193942480236312\\
1.02424599676927	0.19609050509465\\
1.02561423219977	0.198252055226298\\
1.02697152553622	0.200427125497423\\
1.02831780868326	0.202615709939918\\
1.02965301406457	0.204817801748419\\
1.03097707463155	0.207033393277441\\
1.03228992387195	0.209262476038622\\
1.03359149581853	0.211505040698083\\
1.03488172505756	0.213761077073901\\
1.03616054673734	0.21603057413371\\
1.03742789657665	0.218313519992401\\
1.03868371087313	0.220609901909962\\
1.03992792651163	0.222919706289424\\
1.04116048097246	0.225242918674935\\
1.04238131233963	0.227579523749951\\
1.04359035930897	0.229929505335557\\
1.04478756119622	0.232292846388907\\
1.04597285794505	0.234669529001791\\
1.04714619013501	0.237059534399327\\
1.04830749898937	0.239462842938782\\
1.04945672638294	0.241879434108519\\
1.05059381484982	0.244309286527074\\
1.05171870759098	0.246752377942355\\
1.05283134848186	0.249208685230985\\
1.05393168207987	0.251678184397762\\
1.05501965363176	0.254160850575259\\
1.05609520908094	0.256656658023549\\
1.05715829507471	0.259165580130068\\
1.05820885897141	0.26168758940961\\
1.05924684884746	0.264222657504453\\
1.06027221350429	0.266770755184623\\
1.06128490247524	0.269331852348285\\
1.06228486603228	0.271905918022279\\
1.06327205519274	0.274492920362785\\
1.06424642172581	0.277092826656123\\
1.06520791815903	0.279705603319691\\
1.06615649778467	0.28233121590304\\
1.06709211466597	0.284969629089083\\
1.06801472364328	0.287620806695442\\
1.06892428034014	0.290284711675932\\
1.06982074116916	0.292961306122179\\
1.07070406333791	0.295650551265379\\
1.07157420485458	0.298352407478188\\
1.07243112453358	0.301066834276761\\
1.07327478200105	0.303793790322908\\
1.07410513770021	0.306533233426409\\
1.07492215289659	0.309285120547447\\
1.07572578968319	0.312049407799191\\
1.07651601098545	0.314826050450505\\
1.07729278056616	0.317615002928801\\
1.0780560630302	0.320416218823023\\
1.0788058238292	0.323229650886767\\
1.07954202926603	0.326055251041541\\
1.08026464649917	0.328892970380148\\
1.08097364354702	0.331742759170217\\
1.08166898929195	0.334604566857861\\
1.08235065348433	0.337478342071465\\
1.08301860674641	0.340364032625614\\
1.08367282057601	0.34326158552515\\
1.08431326735015	0.346170946969352\\
1.08493992032851	0.349092062356269\\
1.08555275365671	0.352024876287153\\
1.08615174236956	0.354969332571052\\
1.08673686239408	0.357925374229503\\
1.08730809055245	0.360892943501377\\
1.08786540456471	0.363871981847833\\
1.0884087830515	0.36686242995741\\
1.08893820553647	0.369864227751237\\
1.0894536524487	0.372877314388369\\
1.08995510512488	0.375901628271244\\
1.09044254581139	0.378937107051274\\
1.09091595766625	0.381983687634533\\
1.0913753247609	0.385041306187586\\
1.09182063208182	0.388109898143437\\
1.09225186553211	0.39118939820757\\
1.09266901193275	0.394279740364138\\
1.09307205902389	0.397380857882246\\
1.0934609954659	0.400492683322354\\
1.09383581084029	0.403615148542793\\
1.0941964956505	0.406748184706392\\
1.09454304132253	0.409891722287211\\
1.09487544020545	0.413045691077391\\
1.09519368557174	0.416210020194097\\
1.09549777161748	0.419384638086581\\
1.09578769346242	0.422569472543336\\
1.09606344714991	0.425764450699366\\
1.09632502964665	0.428969499043538\\
1.09657243884229	0.432184543426055\\
1.09680567354898	0.435409509066008\\
1.09702473350063	0.438644320559032\\
1.09722961935216	0.441888901885059\\
1.0974203326785	0.445143176416157\\
1.09759687597356	0.448407066924459\\
1.09775925264891	0.451680495590193\\
1.09790746703248	0.454963384009783\\
1.09804152436699	0.458255653204043\\
1.09816143080831	0.46155722362646\\
1.09826719342364	0.464868015171547\\
1.09835882018957	0.468187947183288\\
1.09843631999004	0.471516938463646\\
1.09849970261405	0.474854907281164\\
1.09854897875335	0.478201771379632\\
1.09858415999993	0.481557447986817\\
1.09860525884337	0.484921853823286\\
};
\addlegendentry{$\nu = (0.3,0.2,0.5)$}

\end{axis}
\end{tikzpicture}%
\caption{Rate function $J(t)$ of $\{\frac{1}{n} g_\nu(X^n)\}$, for a distribution over three symbols $\mu = (0.05,0.1,0.85)$.}
\end{figure}

\begin{proof}[Proof of Theorem~\ref{thm:mism_ldp}]
Observe that by Lemma~\ref{lem:mism_same_tilted_family}, for any $\nu* \in \mathcal{T}_\nu^+$ we have $G_{\nu*}(x) = G_{\nu}(x)$ for all $x \in \mathcal{X}$. 
In particular, this holds for $\nu* = \Pi_{\mathcal{T}_\nu}(\mu)$. Therefore, without loss of generality throughout the proof we assume that 
$\nu = \Pi_{\mathcal{T}_\nu}(\mu)$.

Next, note that as $\frac{1}{n}g_{\nu^n}(X^n)$ 
takes values in a compact subset $[0,\log |\mathcal{X}|]$ of $\mathbb{R}$, it is sufficient to prove that the limit below exists and evaluates to the rate function~(see \cite[Section V]{Beirami-IT} for a formal discussion), i.e., 
	\begin{align}
\lim_{\epsilon \downarrow 0} \lim_{n \to \infty} \frac{1}{n}\log  \mathbb{P}_\mu^n\left(  \left|\frac{1}{n} g_{\nu^n}(X^n) - t \right| < \epsilon \right)=  -J(t).
	\end{align}
We proceed with the proof in three separate cases. 

{\it Case (a)}:  We let $t \in (H(\nu),\log |\mathcal{X}|)$, which implies $\alpha(t) \in (0,1)$ by monotonicity of $H(T(\nu,\alpha))$ for non-negative $\alpha$. Note that
\eqref{eq:D_set} and \eqref{eq:B_set} respecitvely imply
\begin{align}
   & \lim_{\epsilon \downarrow 0} \limsup_{n \to \infty} \frac{1}{n} \log \mathbb{P}_{\mu^n}\left( \left| \frac{1}{n} g_\nu(X^n) - H(T(\nu,\alpha(t))) \right| \leq \epsilon \right) \notag \\
   & \hspace{.5em} \leq \lim_{\epsilon \downarrow 0} \limsup_{n \to \infty} \frac{1}{n} \log \mathbb{P}_{\mu^n}(\mathbf{q}_{X^n} \in \mathcal{D}(\nu,\alpha(t),2\epsilon/\alpha(t))), \label{eq:ldp_upb}\\
   & \lim_{\epsilon \downarrow 0} \liminf_{n \to \infty} \frac{1}{n} \log \mathbb{P}_{\mu^n}\left( \left| \frac{1}{n} g_\nu(X^n) - H(T(\nu,\alpha(t))) \right| \leq \epsilon \right) \notag \\
   & \hspace{.5em} \geq \lim_{\epsilon \downarrow 0} \liminf_{n \to \infty} \frac{1}{n} \log \mathbb{P}_{\mu^n}(\mathbf{q}_{X^n} \in  \mathcal{B}(\nu,\alpha(t),\epsilon/\alpha(t))). \label{eq:ldp_lbd}
\end{align}
Thus, it suffices to show that the RHS of \eqref{eq:ldp_upb} and \eqref{eq:ldp_lbd} both evaluate to $-D(\gamma_{\nu, \mu}(t) \| \mu)$. This is done via Sanov's Theorem. Recall that Sanov's Theorem~\cite[Theorem~6.2.10]{Dembo} states that, for a set of distributions $\mathcal{C}$, 
    	\begin{align}
	    - \inf_{\gamma \in \mathrm{int} \mathcal{C}} D(\gamma \| \mu) &\leq \liminf_{n \to \infty} \frac{1}{n} \log \mathbb{P}(\mathbf{q}_{x^n} \in  \mathcal{C}) \notag \\
	    &\leq \limsup_{n \to \infty} \frac{1}{n} \log \mathbb{P}(\mathbf{q}_{x^n} \in \mathcal{C}) \notag \\
	    & \leq - \inf_{\gamma \in \mathrm{cl} \mathcal{C}} D(\gamma \| \mu).
	\end{align}
	To obtain the upper bound, we apply this result to the set $\mathcal{D}(\nu,\alpha(t),2\epsilon/\alpha(t))$. Observing that this holds for any $\epsilon$, and then letting $\epsilon \downarrow 0$, we get that the RHS of \eqref{eq:ldp_upb} is upper bounded 
	\begin{align}
	    - \lim_{\epsilon \downarrow 0} \inf_{\gamma \in \mathrm{cl} \mathcal{D}(\nu,\alpha(t),2\epsilon/\alpha(t))} D(\gamma \| \mu). \label{eq:lim_eps_to_0}
	\end{align}
	We now make use of a basic topological fact. Observe that $D(\gamma \| \mu)$ is strictly convex in $\gamma$ for a fixed $\mu$, and thus there is a unique minimizer $\gamma(t,\epsilon)$. Noting that the minimizer $\gamma(\epsilon,t)$ is in the set $\mathrm{cl}\mathcal{D}(\nu,\alpha(t),2\epsilon/\alpha(t))$, by continuity of $D(\gamma \| \mu)$ and compactness of the set. Thus, the collection of minimizers $\gamma(t,\epsilon)$ is a collection of points such that $\gamma(t,\epsilon) \in \mathrm{cl}\mathcal{D}(\nu,\alpha(t),2\epsilon/\alpha(t))$. It follows from compactness that the limit point $\lim_{\epsilon \downarrow 0} \gamma(\epsilon,t) \in \bigcap_{\epsilon > 0} \mathrm{cl}\mathcal{D}(\nu,\alpha(t),2\epsilon/\alpha(t)) = \mathrm{cl} \mathcal{D}(\nu,\alpha(t),0)$, where we have used that $\alpha(t) > 0$. Therefore, we have the bound
	\begin{align}
	    & \inf_{\gamma \in \Delta_{\mathcal{X}}} \hspace{3em} D(\gamma \| \mu) \notag \\
	    & \text{subject to } \hspace{1.5em} H(\gamma\|\nu) \leq H(T(\nu,\alpha(t))\| \nu)
	\end{align}
	Note that this optimization problem is convex, and thus can be solved analytically by writing the KKT conditions~\cite{boyd2004convex}, which give a solution $\gamma_{\nu,\mu}(t) \in \mathcal{T}_{\nu,\mu}$, and optimal value $D(\gamma_{\nu,\mu}(t)\| \mu)$.
	
	Analogously, the RHS of \eqref{eq:B_set} can be shown to be lower bounded by $- \inf D(\gamma \| \mu)$, where $\gamma \in \mathcal{B}(\nu,\alpha(t),0)$, by noting $\mathcal{B}(\nu,\alpha,0) \subset \mathrm{int} \mathcal{B}(\nu,\alpha,\epsilon)$, for any $\epsilon > 0$. Again, this optimization can be solved analytically, and gives the desired output. Putting these results together, we get that
	\begin{align}
\lim_{\epsilon \downarrow 0} \lim_{n \to \infty} \frac{1}{n}\log  &\mathbb{P}_\mu^n\left(  \left|\frac{1}{n} g_{\nu^n}(X^n) - H(T(\nu,\alpha(t)))\right| < \epsilon \right) \notag \\ 
&= -D(\gamma_{\nu,\mu}(t) \| \mu).
	\end{align}
	
	{\it Case (b):} We now let $t \in (0,H(\nu))$, which implies $\alpha(t) \in (1,\infty)$. The proof in this case follows from the same step as in Case (a), by replacing the set $\mathcal{D}(\nu,\alpha(t),\epsilon)$ with the set $\mathcal{E}(\nu,\alpha(t),\epsilon)$.
	
	{\it Case (c):} Finally, let $t = H(\nu)$, or equivalently, $\alpha(t) = 1$. In this case, note that $\mu \in \mathcal{B}(\nu,1,\epsilon)$, and thus, by the law of large numbers and \eqref{eq:ldp_lbd}, we have that
\begin{align}
    \lim_{\epsilon \downarrow 0} \lim_{n \to \infty} \frac{1}{n}\log  \mathbb{P}_\mu^n\left(  \left|\frac{1}{n} g_{\nu^n}(X^n) - t \right| < \epsilon \right) \geq 0,
\end{align}	
	which implies that $J(t) = 0$ in this case.
	\end{proof}
	
	As mentioned before, an attractive feature of the LDP is that it implies the asymptotic average growth rate of the $\rho$-th moment of the mismatched guesswork, i.e., $E_\rho(\mu \| \nu)$. This is formalized in the following corollary, which is the second main result of this paper implied by Theorem~\ref{thm:mism_ldp}.
	
	\begin{corollary}\label{cor:mism_moments}
	Let $\Pi_{\mathcal{T}_\nu}(\mu) \in \mathcal{T}_\nu^+$. Then, we have
	\begin{align}
	    E_\rho(\nu \| \mu)  & = \max_{\gamma \in \mathcal{T}_{\nu,\mu} } \quad H(\Pi_{T_\nu}(\gamma)) - \frac{1}{\rho} D(\gamma \| \mu) \label{eq:first_form}
	\end{align}
	\end{corollary}

\begin{figure}
	\centering
\begin{tikzpicture}

\begin{axis}[%
width=2.3in,
height=2.3in,
at={(1.273in,0.674in)},
scale only axis,
xmin=0,
xmax=9,
xlabel style={font=\color{white!15!black}},
xlabel={$\rho$},
ymin=0.5,
ymax=1.0986,
ylabel style={font=\color{white!15!black}},
ylabel={$E_\rho(\nu \| \mu)$},
axis background/.style={fill=white},
legend style={at={(0.323,0.04)}, anchor=south west, legend cell align=left, align=left, draw=white!15!black}
]

\addplot [color=mycolor2, line width=1.0pt]
  table[row sep=crcr]{%
0.1	0.552920880484997\\
0.6	0.686882864982004\\
1.1	0.772937820784424\\
1.6	0.830913995455473\\
2.1	0.87134384166833\\
2.6	0.901584326806767\\
3.1	0.923324015453998\\
3.6	0.943383679451936\\
4.1	0.958550742474767\\
4.6	0.970420617883939\\
5.1	0.979963066742293\\
5.6	0.987801506875941\\
6.1	0.996044578226693\\
6.6	1.00324355723149\\
7.1	1.00942859553139\\
7.6	1.01479981300236\\
8.1	1.01950791720531\\
8.6	1.02366856743116\\
9.1	1.02737200334649\\
9.6	1.0306896646873\\
};
\addlegendentry{$\nu = \mu$}

\addplot [color=mycolor3, dashed, line width=1.0pt]
  table[row sep=crcr]{%
0.1	0.601751721850402\\
1.1	0.823952524394913\\
2.1	0.912551430536823\\
3.1	0.954985305395276\\
4.1	0.986122305569564\\
5.1	1.00504871744021\\
6.1	1.01776974836966\\
7.1	1.0277643756821\\
8.1	1.03602570360766\\
9.1	1.04247135506606\\
};
\addlegendentry{$\nu =(0.32,0.3,0.37)$}

\addplot [color=mycolor1, dotted, line width=1.5pt]
  table[row sep=crcr]{%
0.1	0.630142980889312\\
0.6	0.776224012112979\\
1.1	0.85357813822182\\
1.6	0.901706697683335\\
2.1	0.934277573780422\\
2.6	0.957705666792946\\
3.1	0.975332917768435\\
3.6	0.989066390863711\\
4.1	1.00006077545673\\
4.6	1.00905871366116\\
5.1	1.01655541646673\\
5.6	1.02289763484482\\
6.1	1.02833188737549\\
6.6	1.03303964135799\\
7.1	1.03715726614613\\
7.6	1.04079009872694\\
8.1	1.04401745431499\\
8.6	1.04690381682369\\
9.1	1.04949999491734\\
9.6	1.05184984738489\\
};
\addlegendentry{$\nu = (0.3,0.2,0.5)$}

\end{axis}
\end{tikzpicture}%
\caption{Illustration of Corollary~\ref{cor:mism_moments}. The distributions are identical as in Figure~1. Note that, as $\rho$ grows, the curves meet at $\log |\mathcal{X}|$.}
\end{figure}
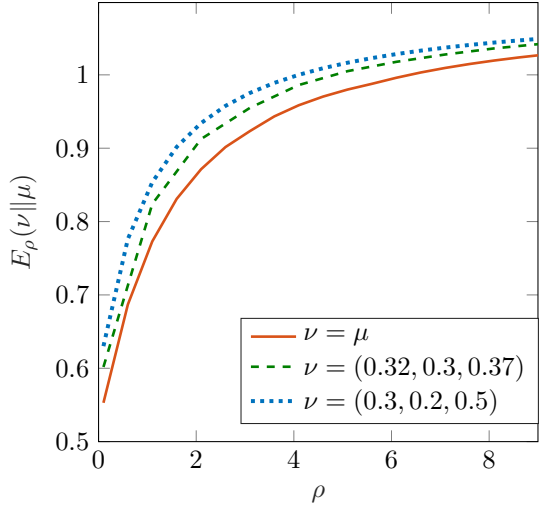	
	\begin{proof}
	    We use Varadhan's Lemma~\cite[Theorem 4.3.1]{Dembo}, which states that if a sequence of random variables $M_n$ satisfies a LDP with rate function $J(t)$, then we have
	    \begin{align}
	    \lim_{n \to \infty} \frac{1}{n} \log \mathbb{E}_{\mu^n}[\exp n F(M_n)] = \sup_{t} F(t) - J(t),
	    \end{align}
	    for any continuous and bounded function $F$. Applying this results to the sequence $\{\frac{1}{n}g_\nu(X^n)\}$, and letting $F(t) = \rho \cdot t$, for $\rho > 0$ and $t \in [0,\log|\mathcal{X}|]$ thus yields
	     \begin{align}
	    \lim_{n \to \infty} \frac{1}{n} \log \mathbb{E}_{\mu^n}[G_\nu^\rho(X^n)] = \sup_{t} \rho\cdot t - D(\gamma_{\nu, \mu}(t) \| \mu).
	    \end{align}
	    Performing the optimization on $\gamma$ instead of $t$, via the change of variables in \eqref{eq:alpha_t} and \eqref{eq:gamma_t} concludes the proof.
	\end{proof}

The following is an immediate corollary which lower bounds the mismatched guesswork.
\begin{corollary}[non-negativity of mismatch penalty]
Let $\Pi_{\mathcal{T}_\nu}(\mu) \in \mathcal{T}_\nu^+$, then the following holds:
\begin{align}
  E_\rho(\nu \| \mu) \geq E_\rho(\mu) = H_{\frac{1}{1+\rho}}(\mu),
\end{align}
with equality iff $\mu \in \mathcal{T}_{\nu}^+$.
\end{corollary}
\begin{proof}
Consider the optimization from \eqref{eq:regular_guesswork}, and notice that it can be equivalently written as
\begin{align}
    \max_{\phi \in \mathcal{T}_{\mu}} \; H(\phi) - \frac{1}{\rho} D(\phi\| \mu)
\end{align}
Using Lemma~\ref{lem:entr_ineq}, we obtain that $H(\Pi_{\mathcal{T}_\nu} (\zeta)) \geq H(\zeta))$, giving the upper bound
\begin{align}
     \max_{\phi \in \mathcal{T}_{\mu}} \; H(\Pi_{\mathcal{T}_\nu} (\phi)) - \frac{1}{\rho} D(\phi\| \mu).
\end{align}
Next, notice that since $H(\Pi_{\mathcal{T}_\nu} (\phi)\| \nu) = H(\phi \| \nu)$, by definition of $\Pi_{\mathcal{T}_\nu}$, it must be the case that $D(\gamma \| \mu) < D(\phi \| \mu)$ for some $\gamma \in \mathcal{T}_{\nu,\mu}$ which satisfies $H(\gamma \| \nu) = H(\Pi_{\mathcal{T}_\nu} (\phi)\| \nu)$. It follows that
\begin{align}
     \max_{\phi \in \mathcal{T}_{\mu}} & \quad H(\Pi_{\mathcal{T}_\nu} (\phi)) - \frac{1}{\rho} D(\gamma \| \mu) \\
     \text{such that} & \quad H(\gamma \| \nu) = H(\Pi_{\mathcal{T}_\nu} (\phi)\| \nu)
\end{align}
is an upper bound to the matched guesswork.
The proof follows from performing the change of variable $H(\Pi_{\mathcal{T}_\nu} (\zeta)) = T(\nu,\alpha)$, and identifying the resulting optimization as being equivalent to \eqref{eq:first_form}.
\end{proof}

\section{Applications to one-to-one Coding} \label{sec:one_to_one}

In this section, we connect the established results to lossless source coding. We follow the notation from \cite{courtade2014cumulant}, and start by a discussion on lossless coding without mismatch.
A lossless source code is an injective function $f: \mathcal{X} \to \{ 0,1\}^*$, and we refer to $f(x)$, for some $x \in \mathcal{X}$ as a codeword. For a codeword $c \in \{ 0,1\}^*$, the length of the codeword is denoted by $l(c)$. A lossless source code $f^*$ is said to be optimal if it satisfies $\mathbb{E}[l(f^*(X))] \geq \mathbb{E}[l(f(X))]$ for all valid source codes $f$.

The relationship between the optimal source code $f^*$ and the log-guesswork $g_\mu$, was discussed in \cite{ArikanMerhav}
\cite{HanawalSundaresan}, and later in \cite{ChristiansenDuffy}. Essentially, this correspondence is due to the relation $\mu(x) \geq \mu(y) \iff l(f^*(x)) \leq l(f^*(y))$, which imposes that there is an optimal encoding with $l(f^*(x)) \geq \lfloor \log_2 G_\mu(x) \rfloor$ for all $x \in \mathcal{X}$. For iid sources, the asymptotic behavior of lossless codes are investigated through two quantities of interest, namely the asymptotic average length, and the reliability function
\begin{align}
    &L(\mu) \triangleq \lim_{n \to \infty} \frac{1}{n} \mathbb{E}[l(f^*(X^n))], \\
    &E(R,\mu) \triangleq - \liminf_{n \to \infty} \frac{1}{n} \log \mathbb{P}_{\mu^n}\left(l(f^*(X^n)) > n R \right) ,
\end{align}
where $H(\mu) <R <\log |\mathcal{X}|$. Naturally, the average length $L(\mu) = H(\mu)$, that is, the best average length for a lossless code is asymptotically converging to the entropy of the source, see \cite{szpankowski-2008}. By using the correspondence between $l(f^*(x^n))$ and $g_\mu(x^n)$, one can directly apply the results in Theorem~\ref{thm:main} to obtain closed forms on the reliability function $E(R,\mu)$ (we refer to \cite{courtade2014cumulant} for more details). In the rest of this section, we discuss analogous quantities for the case of mismatched lossless coding without prefix-free constraint.

Now, assume that an optimal lossless source code is constructed according to a mismatched source statistic $\nu$. We let $f^*_\nu$ be the resulting optimal code for the source statistic $\nu$, and define the asymptotic average length and reliability function similarly as in the matched case, i.e.,
\begin{align}
    & L(\nu \| \mu) = \lim_{n \to \infty} \frac{1}{n} \mathbb{E}[l(f^*_{\nu}(X^n))] \\
    & E(R,\nu \| \mu) \triangleq - \liminf_{n \to \infty} \frac{1}{n} \log \mathbb{P}_{\mu^n}\left(l(f_{\nu}^*(X^n)) > n R \right).
\end{align}
The following is the main result of this section, and is a direct consequence of the LDP result on the mismatched guesswork.
\begin{theorem}
Let $X^n \sim \mu^n$, and assume $\Pi_{\mathcal{T}}(\mu) \in \mathcal{T}_\nu^+$, then:
\begin{align}
    &L(\nu \| \mu) = H(\Pi_{\mathcal{T}_\nu}(\mu)), \\
    &E(R,\nu \| \mu) = J(R),
\end{align}
for $H(\Pi_{\mathcal{T}_\nu}(\mu)) < R < \log |\mathcal{X}|$.
\label{thm:mismatched-one-to-one-coding}
\end{theorem}
\begin{proof}
    The proof of the statement on the reliability function follows immediately by noting that there is an optimal encoding such that $g_\nu(x^n)\leq l(f^*_\nu(x^n)) < g_\nu(x^n) + 1$, and by applying Theorem~\ref{thm:mism_ldp}. The result on $L(\nu \| \mu)$ follows from:
    \begin{align}
        L(\nu \| \mu) &= \lim_{n \to \infty} \frac{1}{n} \mathbb{E}_{\mu^n}\left[ g_{\nu}(X^n)\right] \\
        & = \lim_{\rho \downarrow 0} E_{\rho}(\nu \| \mu),
    \end{align}
    where the first equality is again a consequence of the correspondence between optimal code and guesswork, while the second equality is an application of L'H\^opital's rule. Recall that, by Corollary~\ref{cor:mism_moments}, $E_\rho(\nu \| \mu ) = \max_{\gamma \in \mathcal{T}_\nu^+} H(\Pi_{\mathcal{T}_\nu}(\gamma))) - \frac{1}{\rho} D(\gamma \| \mu)$. It follows that when $\rho \downarrow 0$, it must be that $\gamma = \mu$, which results in $L(\nu \| \mu) = H(\Pi_{\mathcal{T}_\nu}(\mu))$.  
\end{proof}
In prefix free coding, the average length of the coded iid sequence is governed by the cross entropy $H(\mu \| \nu)$, where $\mu$ is the true distribution, and $\nu$ is the mismatched distribution used to generate the code. In particular, since $D(\mu \| \nu) \geq 0$, with equality only if $\mu = \nu$, there is always a loss in performance in using a mismatched distribution. The result above guarantees that the performance of a lossless one-to-one code always exceeds that of a prefix-free code in terms of asymptotic average length, in the presence of mismatch. Indeed, we have by Lemma~\ref{lem:entr_ineq},
\begin{align}
    H(\Pi_{\mathcal{T}_\nu}(\mu)) &= H(\mu\|\Pi_{\mathcal{T}_\nu}(\mu)) \\
    & = H(\mu) + D(\mu \| \Pi_{\mathcal{T}_\nu}(\mu)) \\
    & \leq H(\mu) + D(\mu \|\nu),
\end{align}
where the last step follows from Lemma~\ref{thm:pythagore}.
Therefore, the penalty induced by mismatch from one-to-one coding is always upper bounded by the penalty for prefix-free codes as the asymptotic average codeword length in both cases is characterized by $H(\mu)$~\cite{szpankowski-2008}. The relative entropy $D(\mu \| \Pi_{\mathcal{T}_\nu}(\mu))$ can also be 0, if $\mu \in \mathcal{T}_\nu^+$, i.e., if $\mu$ and $\nu$ are on the same tilted distribution. This implies that the cost of mismatched source coding vanishes if and only if $\mu \in \mathcal{T}_\nu^+$ (Lemma~\ref{lem:mism_same_tilted_family}), and Theorem~\ref{thm:mismatched-one-to-one-coding} generalizes such characterization to arbitrary mismatched distributions.

\section{Conclusion}\label{sec:conclusion}

In this paper, we revisited mismatch guesswork using geometric insights. In particular, we generalized the tilted families of \cite{Beirami-IT}, and showed that the LDP rate function is implicitly expressed in terms of the relative entropy between distributions on this tilted family, and the true distribution $\mu$. We applied these results to the case of one-to-one lossless coding, and showed that, perhaps surprisingly, one-to-one coding is more robust to mismatch than prefix-free coding. Interestingly, similar tilted distributions have appeared in the context of error exponents, see e.g. \cite{borade2006projection}. A more in depth study of the relationship between mismatched guesswork and error exponents for random coding is of future interest.

\section*{Acknowledgment}
The authors would like to thank Robert Calderbank (Duke University), Ken Duffy (National University of Ireland Maynooth), and Wasim Huleihel (Tel Aviv University) for insightful discussions on the topic of mismatched one-to-one source coding and guesswork.

\bibliographystyle{IEEEtran}
\bibliography{strings}

\end{document}